\documentclass[a4paper, 11pt]{article}

\usepackage{graphicx,color}

\usepackage{amsmath,amssymb,amsfonts,amsthm,amssymb,latexsym}

\usepackage[top=2.2 cm,bottom=2.2 cm,left=2.2 cm, right=2.2 cm]{geometry}

\newtheorem{theorem}{Theorem}[section]
\newtheorem{corollary}{Corollary}

\newtheorem{lemma}[theorem]{Lemma}
\newtheorem{proposition}{Proposition}

\theoremstyle{definition}

\newtheorem{remark}{Remark}
\newcommand{\Real}{\mathbb {R}}

\newcommand{\cA}{\mathcal{A}}
\newcommand{\cH}{\mathcal{H}}

\DeclareMathOperator*{\Rel}{Re}

\date{}
\title%[Stability of nonconstant stationary solutions]
{Stability of nonconstant stationary solutions in
a reaction-diffusion equation coupled to the system of ordinary differential equations }

\begin{document}

\maketitle
\vspace{-1 cm }

\centerline{\scshape Yuriy Golovaty}
\smallskip
{\footnotesize
% please put the address of the first author
 \centerline{Department of Mechanics and Mathematics}
   \centerline{Ivan Franko National University of L'viv}
\centerline{ Universytetska str. 1,
L'viv 79000, Ukraine}
\centerline{ yu\_holovaty@franko.lviv.ua}
} % Do not forget to end the {\footnotesize by the sign }

\medskip

% Enter the first author's name and address:
\centerline{\scshape Anna Marciniak-Czochra}
\smallskip
{\footnotesize
 % please put the address of the second  and third author
 \centerline{University of Heidelberg,}
   \centerline{Interdisciplinary Center for Scientific Computing (IWR)}
   \centerline{Institute of Applied Mathematics and BIOQUANT}
      \centerline{Im Neuenheimer Feld 267, 69120 Heidelberg, Germany }
 \centerline{  anna.marciniak@iwr.uni-heidelberg.de }  } % Do not forget to end the {\footnotesize by the sign } % Do not forget to end the

\medskip
\centerline{\scshape Mariya Ptashnyk }
\smallskip
{\footnotesize
 % please put the address of the second  and third author
 \centerline{Department of Mathematics} 
 \centerline{ University of Dundee}
   \centerline{ DD1 4HN Dundee, Scotland, UK}
\centerline{ m.ptashnyk@dundee.ac.uk}}

\bigskip

% The name of the associate editor will be entered by an editorial staff
% \centerline{(Communicated by the associate editor name)}

%The abstract of your paper
\begin{abstract}
In this paper we study pattern formation arising in a system of a single reaction-diffusion equation coupled with subsystem of ordinary
 differential equations, describing spatially-distributed growth of clonal populations of precancerous cells,
whose proliferation is controled by growth factors diffusing in the extracellular medium and binding to the cell surface. We extend the results on the existence of nonhomogenous stationary solutions obtained in \cite{Marciniak} to a general Hill-type production function and full parameter set. Using spectral analysis and perturbation theory we derive conditions for the linearized stability of such spatial patterns.
 \end{abstract}
 
 {\footnotesize
{\it 2000~MSC} {Primary:	35K57, 35J57; Secondary: 92B99.}

 }
\section{Introduction}
Partial differential equations of diffusion type have long served to model regulatory feedbacks and pattern formation in aggregates of living cells. Classical mathematical models of
pattern formation in cell assemblies have been constructed using reaction-diffusion equations. They have been applied to describe pattern formation of animal coat markings, bacterial
and cellular growth patterns, tumor growth and tissue development, see e.g.,  \cite{Murray} and \cite{Maini} and references therein.  One of the mechanisms of pattern formation in reaction-diffusion systems,  prevalent  in the modeling literature since the seminal paper of Allan Turing \cite{Turing}, is diffusion driven instability  (Turing-type instability).

Diffusion-driven instability arises in a reaction-diffusion system, when there exists a spatially homogeneous solution, which is asymptotically
stable in the sense of linearized stability in the space of constant functions, but
it is unstable with respect to spatially inhomogeneous perturbation. The majority of theoretical studies in theory of pattern formation focus on the analysis of the systems of two or more reaction-diffusion equations. In many biological applications it is relevant to consider systems consisting of a single reaction-diffusion equation coupled with a system of ordinary differential equations.  Such models can also exhibit diffusion-driven instability. However, they are very different from classical Turing-type models and the spatial structure of the pattern emerging from the destabilisation of the spatially homogeneous steady state cannot be concluded from a linear stability
analysis \cite{our_CMMM,Murray}. The models exhibit qualitatively
new patterns of behavior of solutions, including, in some cases, a strong dependence of the emerging pattern on initial conditions and quasi-stability followed by rapid growth of solutions \cite{our_CMMM}.

Mathematical theory exists only for special cases of such solutions arising in the models, which can be simplified to the
single reaction-diffusion equation with nonlocal terms and with fast growing kinetics. The example of such systems are
Gierer-Meinhard and Gray-Scott models. The existence and stability of the solutions of such systems were intensively
studied using singular perturbation analysis and spectral analysis of the eigenvalue problem associated with the
linearization around the ``spike-like'' solution, e.g.~\cite{Doelman1,Doelman2,Wei}. The structure of the kinetics
involved in the cell proliferation model considered by us is different, in particular the autocatalysis is excluded by
the biological assumptions and the solutions of our model are uniformly bounded. As shown in \cite{our_M3AS, Marciniak} the models have
stationary solutions of periodic type, the maxima and minima of which may be of the spike or plateau type. Numerical
simulations show that in some cases  solutions of the model converge to a spatially heterogeneous pattern, which persists
for an arbitrary long time while in other cases the transient growth of nonconstant pattern is observed and ultimately
the solution converges to a stable spatially homogeneous state.

In this paper we approach the issue of the stability of nonconstant stationary solutions and investigate linear
stability of spatially heterogeneous solutions of the model of the reaction-diffusion equation coupled with two
ordinary differential equations, which was proposed in \cite{Marciniak}. The paper is organized as follow:
In Section \ref{ModelSection} the model is introduced and the results on existence, regularity  and boundedness of the model solutions are presented.
In Section \ref{StationarySection} existence of a spatially nonconstant steady state is shown. Section \ref{StabilitySection} is devoted
to the linearized stability analysis.

\section{Model description}\label{ModelSection}
We consider a model of a cell population controlled by a diffusive growth factor,
\begin{equation}\label{main}
\begin{cases}
\: \partial_t c = \dfrac {\theta b \, c}{b+c}  - d_c c+ \mu &\text{ in } (0,\infty)\times(0,1), \\
\: \partial_t b =\alpha c^2 g - d_b b - \nu b  & \text{ in } (0,\infty)\times(0,1), \\
\: \partial_t g = \dfrac 1 \gamma \,\partial_x^2 g - \alpha c^2 g - d_g g + \dfrac{\beta c^k}{1+c^k}+\nu b
& \text{ in } (0,\infty)\times(0,1),  \\
\: \partial_x g(t,0)=\partial_x g(t,1)=0   & \text{ in } (0,\infty),
\end{cases}
\end{equation}
with initial conditions
\begin{eqnarray*}\label{main_initCond}
c(0,x)= c_{in}(x), \quad  b(0,x)= b_{in}(x), \quad  g(0,x)= g_{in}(x) \quad \text{ for } x\in (0,1),
\end{eqnarray*}
where $c$ denotes the concentration of precancerous cells, whose proliferation rate  is reduced by cell crowding but enhanced in
a paracrine manner by a hypothetical biomolecular growth factor
$b$ bound to cells. Free growth factor $g$ is secreted by the cells,
then it diffuses among cells with diffusion constant $1/\gamma$, and binds to cell membrane receptors becoming the bound factor $b$.
Then, it dissociates at a rate $\nu$, returning to the $b$-pool. Parameter $\mu$ denotes a small influx of new precancerous cells due to mutation.
All coefficients in the system \eqref{main} are assumed to be constant and positive, and $k\in \mathbb N$. For  $k$ larger than 1, production of growth factor  molecules is given by a Hill function and models a process with fast the saturation effect.

\begin{theorem}
For nonnegative initial data $(c_{in},b_{in},g_{in}) \in C^{2+\alpha}([0,1])^3$  there exists a nonnegative global solution of
 system \eqref{main},
$(c,b,g)\in C^{1+\alpha/2, 2+\alpha}([0,\infty)\times[0,1])^3$.
\end{theorem}

\begin{proof}
Using the existence and regularity theory for systems of parabolic and ordinary differential equations
(see \cite{Henry}, \cite{Rothe}),
for $(c_{in},b_{in},g_{in}) \in C^{2+\alpha}([0,1])^3$, we obtain, due to local Lipschitz continuity of reaction terms in the system,  the  existence of a local
 solution of \eqref{main},
$(c,b,g)\in C^{1+\alpha/2, 2+\alpha}([0,T_0]\times [0,1])^3$ for some $T_0 <\infty$.

The theory of
bounded invariant rectangles (see \cite{Chueh}, \cite{Smoller}) and the properties of functions
$F_c=\frac {\theta b \, c}{b+c}  - d_c c+ \mu$, $F_b=\alpha c^2 g - d_b b - \nu b $, $F_g=\frac{\beta c^k}{1+c^k}- \alpha c^2 g - d_g g  +\nu b$, i.e.
are smooth for $c\geq 0, b\geq 0, g\geq 0$ and in $\{(c,b,g)\in \mathbb R^3: c\neq - b \}$,
$F_c(\frac{\mu}{d_c},b,g) \geq 0$ for all $b\geq 0$, $g \geq 0$,
$F_b(c,0,g) \geq 0$ for all $c\geq \frac{\mu}{d_c}$, $g \geq 0$,
and
$F_g(c,b,0) \geq 0$ for all $c\geq \frac{\mu}{d_c}$, $b\geq 0$,
imply that the set $\Sigma=\{(c,b,g): c(t,x) \geq \min\{\min_{x\in [0,1]} c_{in}(x), \frac{\mu}{d_c}\}, b(t,x)\geq 0, g(t,x)\geq 0, x \in [0,1], t\geq 0 \}$
is positive invariant and  solutions are nonnegative for nonnegative initial data.
Then, from the first equation in \eqref{main} follows the estimate
$$
\partial_t c \leq  (\theta  - d_c) c+ \mu,
$$
and, by Gronwall inequality, we obtain, that
$c$ is bounded  for all finite time, i.e., $\sup\limits_{x\in[0,1]} c(t,x) \leq  \mu \exp{((\theta  - d_c)t)}$.
The boundedness of $c$ implies also the boundedness of $b$ and $g$, and  existence of a global solution.
\end{proof}

\section{Existence of positive non-constant steady states}\label{StationarySection}

Stationary solutions of \eqref{main} can be found from the system
\begin{equation}\label{eq_stead_state}
\begin{cases}
\;
  \dfrac {\theta  b\, c }{b+c}  - d_c c+ \mu=0,
 \\
\; \alpha c^2 g - d_b b - \nu b=0,\\
\; \dfrac 1 \gamma \, g'' - \alpha c^2 g - d_g g + \dfrac{\beta c^k}{1+c^k}+ \nu b=0, \quad
g'(0) = g'(1)=0.
\end{cases}
\end{equation}
We look for a positive solution $(c(x),b(x),g(x))$ defined in $(0,1)$ that has at least one non-constant component.
It is clear that for any solution of \eqref{eq_stead_state} either all functions $c$, $b$ and $g$ are constant or all depend on $x$ effectively. The next observation is that the function $c$ doesn't change the sign. In fact, if $c(x_0)=0$ for some $x_0\in (0,1)$, then it follows from the first equation that $\mu=0$, which is impossible.
Let us solve the system
 \begin{equation}\label{eq_stead_statePart}
\begin{cases}
\;
  \dfrac {\theta  b\, c }{b+c}  - d_c c+ \mu=0,
 \\
\; \alpha c^2 g - d_b b - \nu b=0
\end{cases}
\end{equation}
with respect to $c$ and $b$.

\begin{lemma}\label{LemmaC}
  If $\theta <d_c$, then there exists a unique  solution $(c(g),b(g))$ of \eqref{eq_stead_statePart},
which is continuous and positive for all $g\geq0$. In the case $\theta = d_c$ there is only one positive solution and
for $\theta>d_c$  there are two positive solutions  that exist only in some interval $(0,g^*)$, where $g^*$ is a finite number.
\end{lemma}
\begin{proof}
  Let us introduce the temporary notation $\ell(g)=\frac{\alpha \mu }{d_b +\nu}\, g$. First observe that
 \begin{equation}\label{B_as_functin_c}
    b(g)=\frac{1}{\mu}\,\ell(g)\, c^2(g),
 \end{equation}
 and hence $b$ is positive as soon as $g$ is positive.
  The main question is whether there in fact exists a positive $c$.
  Substituting expression  \eqref{B_as_functin_c} into the first equation of \eqref{eq_stead_statePart} yields
\begin{equation}\label{quad_eq_c}
   (\theta-d_c)\ell \, c^2 + \mu(\ell-d_c) \, c + \mu^2=0.
\end{equation}
This quadratic equation with respect to $c$ admits  real roots provided 
\begin{equation*}
    D=\mu^2\left((\ell-d_c )^2-4(\theta -d_c)\ell\right)=\mu^2\left((\ell+d_c)^2-4 \frac{\theta }{d_c}\,d_c\ell\right)\geq 0.
\end{equation*}
Note that $(p+q)^2-4\varepsilon pq\geq 0$ for all positive $p$, $q$ if and only if $\varepsilon\leq 1$. Hence,
under the assumption $\theta <d_c$, equation \eqref{quad_eq_c} has  the  solution
\begin{equation}\label{PositiveCg}
 c(g)=\mu \cdot \frac{\ell(g)-d_c+\sqrt{\left(\ell(g)-d_c\right)^2+4(d_c-\theta )\ell(g)}}{2(d_c-\theta )\ell(g)},
\end{equation}
which is positive for all $g>0$. Furthermore,
\begin{equation}\label{limitsOfC}
    c(0+)=\frac{\mu}{d_c},\qquad c(+\infty)=\frac{\mu}{d_c-\theta}.
\end{equation}
In the the critical case $\theta= d_c$ there exists a unique solution of \eqref{quad_eq_c}
\begin{equation}\label{LinearSolutionC}
 c(g)=\frac{\mu}{d_c-\ell(g) },
\end{equation}
where $c$ is positive for $g\in \left[0,\frac{d_c(d_b +\nu)}{\alpha\mu}\right)$. In the case $\theta>d_c$
there are two positive solutions of \eqref{quad_eq_c}, defined for $g\in (0,g^*)$ with
$g^*=\frac{d_b +\nu}{\alpha\mu}\bigl(\sqrt{\theta}-\sqrt{\theta-d_c}\,\bigr)^2$,
 $g^\ast < \frac{d_c(d_b +\nu)}{\alpha\mu}$. But only one of them  is bounded at $g=0$,
and this solution is given by \eqref{PositiveCg}. Moreover the solution is also bounded at the point $g^*$:
\begin{equation}\label{c(g*)}
    c(g^*)= \frac{\mu(d_c-\ell(g^*))}{2(\theta-d_c)\,\ell(g^*)}.
\end{equation}
Note that $g^*$ is the smallest positive point in which the  discriminant $D=D(g)$ vanishes.
Observe also that $\ell(g)<d_c$ for all $g\leq g^*$.
\end{proof}

We next come back to system \eqref{eq_stead_state}. Let $\omega=\frac{\alpha d_b}{d_b+\nu}$.
In view of Lemma~\ref{LemmaC}  the function $g$ must be a positive solution of the nonlinear boundary value problem
\begin{equation}\label{NBVP}
\begin{cases}
      \dfrac 1 \gamma \, g'' =  g \left(d_g + \omega c^2(g)\right) -  \dfrac{\beta c^k(g)}{1+c^k(g)}\quad \text{ in } (0,1),\\
       g'(0) =0,\quad g'(1)=0,
\end{cases}
\end{equation}
where $c$ is given by \eqref{PositiveCg}, if $\theta\neq d_c$, and by \eqref{LinearSolutionC} otherwise.
We must keep in mind that the function $c$  depends on parameters $\alpha$, $\mu$, $\nu$, $\theta$, $d_b$ and $d_c$.

\subsection{Nonlinear boundary value problem}
Let us consider the boundary value problem
\begin{equation}\label{BVPg}
    g''=\gamma h(g),\quad x\in (0,1),\qquad g'(0)=0, \quad g'(1)=0,
\end{equation}
where
$$
h(g)=  d_g g \left( 1 + \omega\, c^2(g)\right) -  \dfrac{\beta c^k(g)}{1+c^k(g)},\qquad \omega=\frac{\alpha d_b }{d_g(d_b+\nu)},
$$
and $c$ is given by \eqref{PositiveCg} or \eqref{LinearSolutionC}.
The equation of this type, so called \textit{system with one degree of freedom}, was intensively studied in classical mechanics. For a deeper analysis of the system we refer to \cite[Sec.12]{ArnoldODE}.

\begin{figure}[htbp]
\centerline{ \includegraphics[width=12cm]{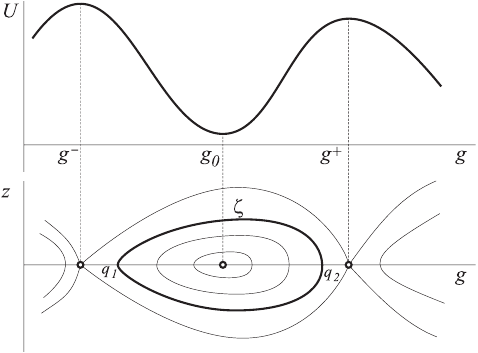}}
\caption{Plots of the energy $U$ and the trajectories of a dyna\-mi\-cal system.}
\label{PicPlotU}
\end{figure}

\begin{theorem}\label{exist_stat} If function $h=h(g)$ has no less than three positive roots, then there exists 
a set $\Gamma\subset\Real_+$ of diffusion constants $\gamma$ for which  boundary value problem~\eqref{BVPg} admits a positive solution.
\end{theorem}

\begin{proof} We assume for a while that $\gamma=1$. Suppose that the potential energy
$$
U(g)=-\int_0^g h(\xi)\,d\xi
$$
has a positive critical point $g_0$  that is a  local minimum of $U$.
The trajectories of dynamic system $g'=z$, $z'=h(g)$ resemble ellipses near the point
$(g_0,0)$ of the phase space. Let us consider the trajectory $\zeta$ that intersects the axis  $z=0$ at positive points $q_1$ and $q_2$ (see Fig.~\ref{PicPlotU}). 
Therefore there exists  a periodic solution  $g=g(x)$, $z=z(x)$ of
 the dynamic system subject to initial conditions $g(0)=q_1$, $z(0)=0$, and its period  is given by
\begin{equation*}
    T=2 \int_{q_1}^{q_2}\frac{d\eta}{\sqrt{2\left(U(q_1)-U(\eta)\right)}}.
\end{equation*}
 Remark that the energy $U$ without local minimum points is either monotonic or a function with a single maximum point. In both cases no trajectory starting at a point $(g_1,0)$ can return again to the line $z=0$, with the exception of the equilibrium positions.

Next,  $g$ is a positive solution of the initial problem $g''=h(g)$, $g(0)=q_1$, $g'(0)=0$. Moreover,
$g'(\frac{nT}2)=0$ for any natural $n$, because $g'(\frac{nT}2)=z(0)=0$ for even $n$ and $g'(\frac{nT}2)=z(\frac{T}2)=0$ for odd $n$. Given $n$, we consider the function $u(x)=g(\frac{x}{\sqrt{\gamma}})$ that is obviously the solution to equation $u''=\gamma h(u)$, and $u'(0)=0$. 
We set $\gamma_n=\frac{4}{n^2T^2}$, so that
$$
u'(1)=\frac{1}{\sqrt{\gamma_n}}\,g'\left(\frac{1}{\sqrt{\gamma_n}}\right)
=\frac{1}{\sqrt{\gamma_n}}\,g'\left(\frac{nT}{2}\right)=0.
$$
Hence, $u$ is a positive solution to \eqref{BVPg}.

It follows that any close trajectory $\zeta$ lying in the half-plane $g>0$ produces  a countable set of positive solutions
$u_n(\zeta,\cdot)$ to \eqref{BVPg} with $\gamma=\gamma_n(\zeta)$. All these sequences $\{\gamma_n(\zeta)\}_{n\in \mathbb{N}}$ form the set $\Gamma$, which is in general uncountable.

The first and primary question is whether the energy $U$ has a local minimum.
We consider first the case
$\theta <d_c$.
The energy has a local minimum at point $g=g_0$ if only $h(g_0)=0$, $h>0$ in a left-side neighborhood of $g_0$ and $h<0$ in a right-side one. Let
$$
 h_1(g)= d_g g \left( 1 + \omega\, c^2(g)\right),\quad h_2(g)=\dfrac{\beta c^k(g)}{1+c^k(g)}.
$$
We look for a root $g_0$ of the equation $h_1(g)=h_2(g)$ for which $h_1(g)>h_2(g)$ the left of $g_0$ and  $h_1(g)<h_2(g)$
on its right.
A trivial verification shows that  $c=c(g)$ (given by \eqref{PositiveCg}) is  a steadily increasing function and
\begin{equation}\label{limitsC}
c(+0)=\frac{\mu}{d_c}, \qquad c(g)\to \frac{\mu}{d_c-\theta}\quad \text{as  } g\to +\infty.
\end{equation}
Set $c_{\rm min}=\frac{\mu}{d_c}$ and $c_{\rm max}=\frac{\mu}{d_c-\theta}$.
It is convenient to consider the $c$-representation  of functions $h_j$ on interval $(c_{\rm min},c_{\rm max})$:
\begin{equation*}
    h_1(c)=\frac{a(c-c_{\rm min})(1+ \omega c^2)}{c(c_{\rm max}-c)},\qquad h_2(c)=\dfrac{\beta c^k}{1+c^k},
\end{equation*}
where $a=\frac{d_c d_g (d_b+\nu)}{\alpha(d_c-\theta)}$.  Both functions
are strictly increasing in the region under study. In regard to $h_1$, its derivative can be written as
\begin{equation*}
     h_1'(c)=a\cdot \frac{(1+ \omega c^2)(c_{\rm min}(c_{\rm max}-c)+c(c-c_{\rm min}))+2 \omega
c^2(c-c_{\rm min})(c_{\rm max}-c)}{c^2(c_{\rm max}-c)^2}
\end{equation*}
and is obviously positive for $c\in (c_{\rm min},c_{\rm max})$. The function $h_1$ has two vertical asymptotes
$c=0$ and $c=c_{\rm max}$, while $h_2$ is bounded.
 In addition, $h_2$ has an inflection point {$c=\sqrt[k]{\frac{k-1}{k+1}}$},
 whereas  the inflection point of $h_1$ depends on parameters of the model.
Fig.~\ref{PicPlotsH1H2} shows the typical plots of $h_1$ and $h_2$. Evidently, the energy $U$ has a local minimum if plots of the
functions intersect at three points.

Similarly,
for $\theta=d_c$ we obtain that $c(g)$  is monoton increasing and
\begin{equation*}
c(+0)= \frac \mu d_c , \qquad c(g) \to \infty \quad \text{ as }  g \to  \frac {d_c(d_b+ \nu)}{\alpha \mu}.
\end{equation*}
The function  $h_1$
is nearly linear   for small $g$ and growth quadratically as $g\to \frac {d_c(d_b+ \nu)}{\alpha \mu}$, 
whereas $h_2(g)\sim c g^k$ for  $g\in(0, \frac {d_c(d_b+ \nu)}{\alpha \mu}) $ such that $c(g)\leq \sqrt[k]{\frac{k-1}{k+1}}$,
 and is  bounded for
$g \to  \frac {d_c(d_b+ \nu)}{\alpha \mu}$;
 additionally,
$h_1(0)=0 < h_2(0)=\frac{\beta (\mu/d_c)^k}{1+(\mu/d_c)^k} $.
Thus, for $k\geq 2$ there exist sets of parameters $d_c, d_b, \nu, \alpha, \mu, \beta$,  such that the function $h(g)$ has three roots in $(0, \frac {d_c(d_b+ \nu)}{\alpha \mu})$.

\begin{figure}[htbp]
\centerline{\includegraphics[width=14.2 cm]{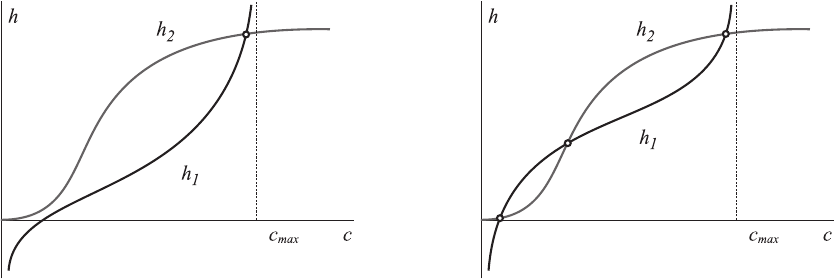}}
\caption{Sample plots of  $h_1$ and $h_2$ corresponding to  nonexistence and existence of a
 nonhomogeneous stationary solution in the left and right panels, respectively.}
\label{PicPlotsH1H2}
\end{figure}

For  $\theta>d_c$ and $c(g)$  defined by \eqref{PositiveCg}, we obtain
\begin{eqnarray*}
 c(+0)=\frac \mu {d_c}, \quad c(g^\ast)= \frac{\mu(d_c- l(g^\ast))}{2(\theta- d_c) l(g^\ast)},
\end{eqnarray*}
 and again $h_1$ is nearly linear,  $h_2(g) \sim c g^k$  for small $g$,   $h_1(0)=0 < h_2(0)$,  and both functions are bounded on $(0, g^\ast)$.
For  $k\geq 2$  there exist two roots, as  functions of parameters,  of $h(g)$ on the interval $(0, g^\ast)$.
The third root of $h(g)$ lies on the interval $(g^{\ast \ast}, \infty)$, where $g^{\ast \ast}$ is the largest positive point in which $D=D(g)$ vanishes.
\end{proof}

\begin{remark}
Let $d_c<\theta$, the solution $c(g)$ of  \eqref{quad_eq_c} is defined by
\begin{equation}\label{sol_c_unb}
 c(g)=\mu \cdot \frac{\ell(g)-d_c-\sqrt{\left(\ell(g)-d_c\right)^2+4(d_c-\theta )\ell(g)}}{2(d_c-\theta )\ell(g)},
\end{equation}
 parameters $d_c$, $d_b$, $\mu$,  $\nu$,  $\alpha$   satisfy the  assumption
 $ 0<g^\ast < (d_b +\nu)d_c/(\alpha \mu)$,  and  $\beta$ is choosen such that $h(g^\ast) <0$.
Then   the function $h$ is continuous  in $(0, g^\ast]$, and $h(g) \to +\infty$ as $g\to +0$. This properties of $h$ and
  the assumption
 $h(g^\ast)<0$ provide existence of a local minimum of the energy  $U$ at a point $g_0 \in (0, g^\ast)$.
Thus, there
  exists  a set
  $\Gamma\subset\Real_+$ of diffusion constants $\gamma$ for which the boundary value problem~\eqref{BVPg}, with $c(g)$
defined by \eqref{sol_c_unb},
admits a positive solution. This situation was considered in \cite{Marciniak} with $k=1$.
\end{remark}
If $U$ has a local minimum $g_0$, then there are a local maximum $g^-$ in the interval $(0,g_0)$ and a local maximum 
$g^+$ in the interval $(g_0, \infty)$  as shown in Fig.~\ref{PicPlotU}.
Set $\varepsilon_0=\min\{|g_0-g^-|, |g_0-g^+|\}$.
\begin{corollary}\label{small_amp}
  For any $\varepsilon\in (0,\varepsilon_0)$ there exists a countable set of solutions $g_n$ to boundary value problem \eqref{BVPg} with $\gamma=\gamma_n$ and  $\gamma_n\to 0$ as $n\to \infty$, such that
  \begin{equation*}
    |g_n(x)-g_0|\leq \varepsilon
  \end{equation*}
for all $x\in (0,1)$.  Moreover $g_n$ is rapidly  oscillating function as $\gamma_n$ goes to $0$.
\end{corollary}

\begin{figure}[htbp]
\centerline{\includegraphics[width=7.1cm]{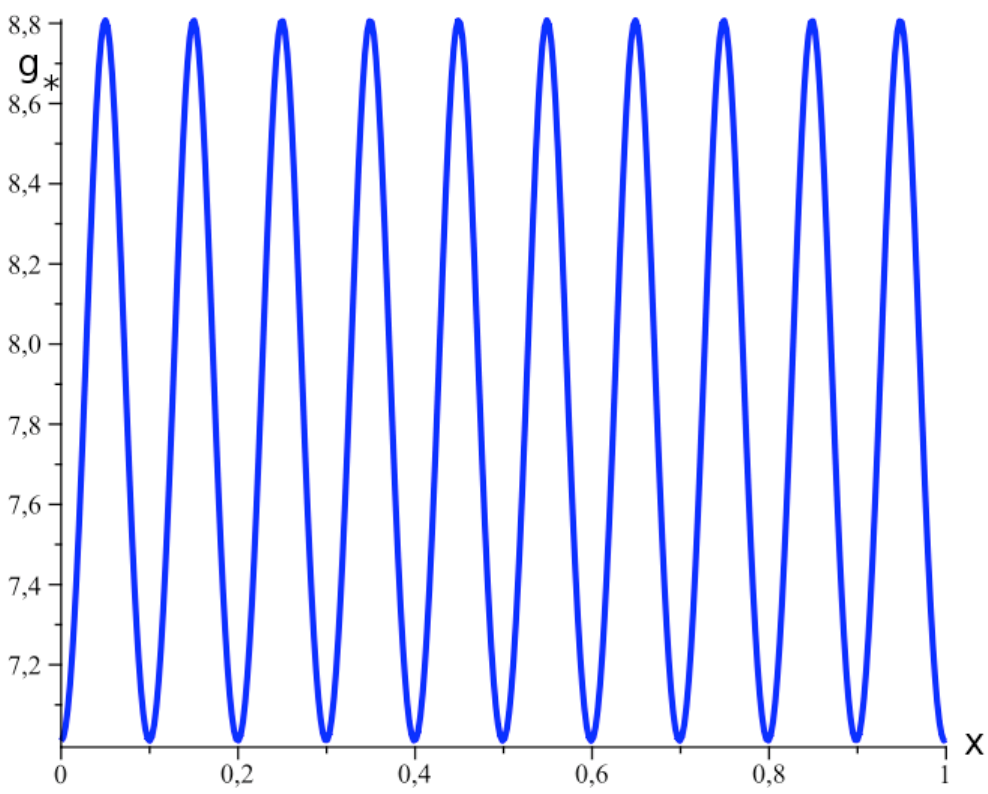} \quad
\includegraphics[width=6.1cm]{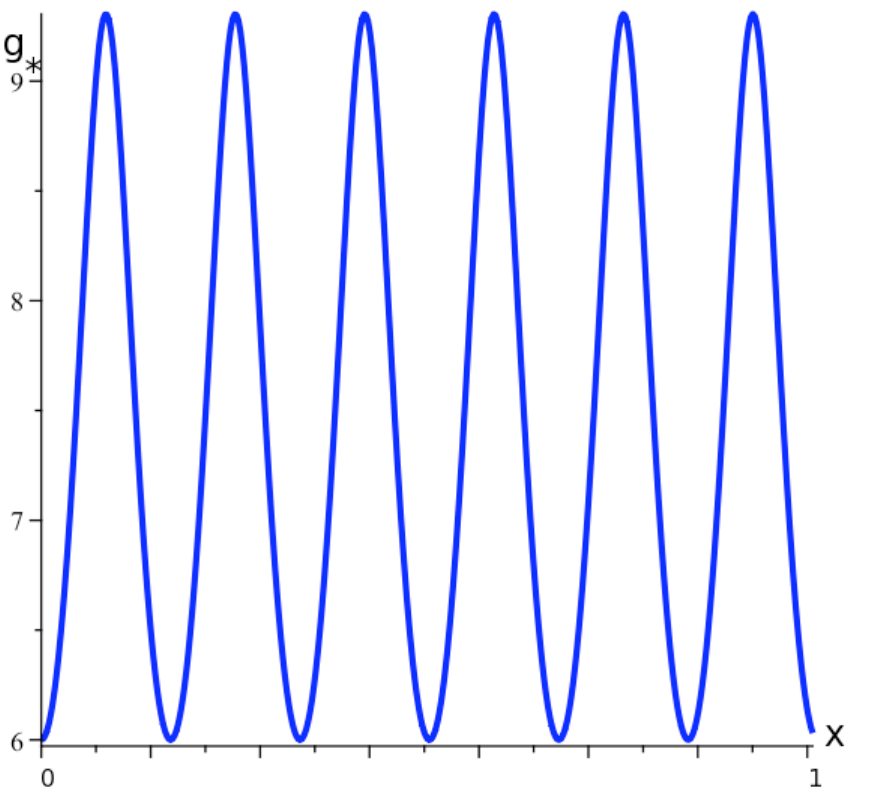}}
\caption{Spatially nonhomogeneous solutions of the stationary problem for the different values of $\gamma$.  
Parameters:  $\alpha=0.1$, $\beta=10$, $\nu=0.1$, $\mu=0.1$, $d_b=d_g=0.1$, $d_c=0.6$,  $\theta=0.59$,
 and $\gamma=12188.16$ and   $\gamma=4900$ in the left and right panels, respectively.}
\label{g_stat}
\end{figure}

\section{Stability and destabilisation of steady states}\label{StabilitySection}

Let $X$ be a complex Banach space with  norm $\|\cdot\|$ and let $A$ be a closed operator on $X$
with a dense domain $\mathcal{D}(A)$. Suppose that $-A$ is a sectorial operator and
$\Rel \lambda<0$ for all $\lambda\in \sigma(A)$.
Hereafter, $\sigma(T)$ stands for the spectrum of an operator $T$.
 For each  $s \geq 0$ we introduce the interpolation space $X^s=\mathcal{D}((-A)^s)$ equipped with the  norm $\|x\|_{s}=\|(-A)^s x\|$. Clearly
$X^0 =X$.
We consider the equation
\begin{equation}\label{help2}
 \partial_t  u=Au+f(u)
\end{equation}
in the Banach space $X$. Let $u^\ast \in \mathcal{D}(A)$ be a stationary solution of (\ref{help2}) .

To study  stability of the  stationary solutions we apply the following proposition
about stability and instability by the linear approximation that is a version
of Theorems~5.1.1 and 5.1.3  in \cite{Henry} adapted for our purposes.
\begin{proposition}\label{Stabil2}
 Let $f:\mathcal{U}\to X$ be a locally Lipschitz continuous map in a neighborhood
 $\mathcal{U}\subset X^s$ of a steady state $u^\ast$, for some $s\in (0,1)$.
Suppose that for $z\in X^s$ the map $f$ admits the representation
\begin{equation}\label{Flinearization}
    f(u^\ast+z)=f(u^\ast)+Bz+p(u^\ast,z),
\end{equation}
provided $\|z\|_s$ is small enough. If
\begin{itemize}
  \item[$\circ$] $B$ is a linear bounded operator from $X^s$ to $X$,
  \item[$\circ$] $||p(u^\ast,z)||=o(||z||_s)$ as  $||z||_s \to 0$,
  \item[$\circ$] the spectrum of $A+B$ lies in the set $\{\lambda\in \mathbb{C}\colon \Rel \lambda < h\}$ for some $h<0$,
\end{itemize}
then the equilibrium  solution $u^\ast$ of (\ref{help2}) is asymptotically stable in $X^s$.
Moreover, if the spectrum  $\sigma(A+B)$ and the half-plane $\{\lambda\in \mathbb{C}\colon \Rel \lambda >0\}$
have a non-empty intersection, then $u^\ast$ is unstable.
\end{proposition}

\subsection{Linearised analysis}
For simplicity of notation, we denote the vector $(c,b,g)$ by $u=(u_1,u_2,u_3)$ so that
system \eqref{main}  reads
\begin{equation}\label{mainVect}
   \partial_t u=Au+f(u)
\end{equation}
in the Banach space $X=C[0,1]\oplus C[0,1]\oplus L^2(0,1)$, where
\begin{equation*}
    A=\begin{pmatrix}
      -d_c& 0 & 0\\
      0& -(d_b+\nu) & 0\\
      0 & 0 &  N_0
    \end{pmatrix},\qquad
   f(u)=\begin{pmatrix}
      \dfrac{\theta u_1 u_2}{u_1+u_2}+\mu&\\
       \alpha u_1^2 u_3& \\
      \dfrac{\beta u_1^k}{1+u_1^k}+\nu u_2-\alpha u_1^2 u_3
    \end{pmatrix}
\end{equation*}
and $N_0$ is  the Sturm-Liouville operator given by $$N_0 v=\frac 1 \gamma \frac{d^2 v}{dx^2}-d_g v$$
on the interval $[0, 1]$, subject to  the Neumann boundary conditions,
$$\mathcal{D}(N_0)=\{v\in H^2(0,1)\colon {v}'(0)= 0, {v}'(1)=0\}.$$
We have that $A$ is a closed densely defined operator on $X$ and  the spectrum of $A$ is given by 
\begin{equation*}
    \sigma(A)=\{-d_c,-d_b-\nu\}\cup\{-d_g-\gamma^{-1}\pi^2j^2\}_{j=0}^\infty.
\end{equation*}
Also  for  $\lambda \notin \sigma(A)$ we have  the estimate 
$$
\|(A- \lambda E)\|_X \leq \frac 1 { \text{dist}(\lambda, \sigma(A))}.
$$
Hence $A$ is a sectorial operator, see e.g.\ \cite{Henry}, and $\Rel \lambda<0$ for all $\lambda\in \sigma(A)$.

For a sectorial operator we can consider interpolation spaces
$X^s=\mathcal{D}((-A)^s)$, for $s\in (0,1)$,  each of which is a Banach subspace of $C[0,1]\oplus C[0,1]\oplus H^{2s}(0,1)$, see e.g.\  \cite{Henry}.
Set $||p||_s=||p_1||_{C[0,1]}+||p_2||_{C[0,1]}+||p_3||_{H^{2s}(0,1)}$.

The function $f$ is smooth in  $\Real_+^3=\{y\in \Real^3\colon y_k>0, \:k=1,2,3\}$. Therefore, $f$ admits the representation
\begin{equation*}
    f(y+z)=f(y)+B(y)z+p(y,z),
\end{equation*}
where the remainder satisfies the estimate
\begin{equation}\label{FTaylor}
 \|p(y,z)\|_{\Real^3}\leq \vartheta(y) \|z\|_{\Real^3}^2
\end{equation}
 in a neighborhood of any point $y\in \Real_+^3$ with a continuous function $\vartheta$. Here
\begin{equation*}
    B(y)=\begin{pmatrix}
     \dfrac {\theta y_2^2}{(y_1+y_2)^2}  &  \dfrac {\theta y_1^2}{(y_1+y_2)^2}    & 0\\
   2\alpha y_1 y_3   & 0 &\alpha y_1^2 \\
    \dfrac{k\beta y_1^{k-1}}{(1+y_1^k)^2}  - 2 \alpha y_1 y_3  & \nu &  -\alpha y_1^2
    \end{pmatrix}.
\end{equation*}

Suppose now that $u^*$ is a positive steady state, that is $Au^*+f(u^*)=0$ and $u^*(x)\in \Real_+^3$ for all $x\in(0,1)$.
Assume also  that $u^*_j$ are $C^2$-functions, $j=1,2,3$.
Obviously, $B(u^*)$ is a linear bounded operator from $X^s$ to $X$, for each $s\in (0,1)$.

Next, using  \eqref{FTaylor} we obtain
$f(u^*+z)=f(u^*)+B(u^*)z+p(u^*,z)$
with the estimate of the remainder
\begin{equation*}
\|p(u^*,z)\|_{\Real^3}\leq \vartheta(u^*(x)) \,\|z(t,x)\|_{\Real^3}^2 \leq  \max_{x\in [0,1]}\vartheta(u^*(x))\|z(t,x)\|_{\Real^3}^2\leq \vartheta_{max} \|z\|_{\Real^3}^2.
\end{equation*}
%and 
%\begin{equation*}
%\begin{aligned}
%&\left\|\frac d{dx} p(u^*,z)\right\|_{\Real^3}\leq \vartheta\left(u^*(x), \frac d{dx} u^\ast (x)\right) \,\|z(t,x)\|_{\Real^3} \left\|\frac d{dx} z(t,x)\right\|_{\Real^3} \\ 
%&\leq  \max_{x\in [0,1]}\vartheta\left(u^*(x), \frac d{dx} u^\ast (x)\right)\|z(t,x)\|_{\Real^3} \left\|\frac d{dx} z(t,x)\right\|_{\Real^3}\leq \vartheta^1_{max} \|z\|_{\Real^3} \left\|\frac d{dx} z\right \|_{\Real^3}
%\end{aligned}
%\end{equation*}

Consequently, for any $z\in X^s$ using the fact that $\|v\|_{C[0,1]} \leq C \|v\|_{H^1(0,1)}$ we conclude
\begin{multline*}
    \|p(u^*,z)\|_X\leq c_1 \bigl(\|z_1\|^2_{C[0,1]}+||z_2||^2_{C[0,1]}+||z_3||^2_{H^1(0,1)}\bigr)\\
    \leq  c_1 \bigl(\|z_1\|^2_{C[0,1]}+||z_2||^2_{C[0,1]}+||z_3||^2_{H^{2s}(0,1)}\bigr)\leq c_2 \|z\|^2_s=o(\|z\|_s),
\end{multline*}
as $\|z\|_s\to 0$, with constants $c_j$, where $j=1,2$, being independent on $t$, and $s\in [1/2, 1)$.

For nonlinear model~\eqref{main} the linearized system  at the steady state $u^*$ can be written as
$\partial_t z= (A+B(u^*))z$. Then turning back to the previous notation, we obtain
\begin{equation}\label{main_lin}
\begin{cases}
\; \partial_t z_1 =\biggl(\dfrac {\theta b_\ast^2}{(c_\ast + b_\ast)^2}- d_c \biggr)\, z_1  +
 \dfrac {\theta  c_\ast^2}{(c_\ast+ b_\ast)^2}\, z_2, \\
\;\partial_t z_2 =
2\alpha c_\ast g_\ast \, z_1  - (d_b+\nu)\,  z_2 +
 \alpha c_\ast^2 \, z_3, \\
\;\partial_t z_3 = \biggl(\dfrac{k\beta c_\ast^{k-1}}{(1+c_\ast^k)^2} - 2 \alpha c_\ast g_\ast\biggr)\, z_1 + \nu\, z_2
  + \biggl(\dfrac 1\gamma \partial_x^2-\alpha c_\ast^2  - d_g \biggr)\,z_3.
\end{cases}
\end{equation}
The first and primary question is whether the spectrum of $A+B(u^*)$ lies  in the half-plane
$\{\lambda\in \mathbb{C}\colon \Rel \lambda < h\}$ for some $h<0$.

\subsection{Spectral analysis of  linearised problem}
Let us consider  the Sturm-Liouville operator  $N =\dfrac 1 \gamma \dfrac{d^2 }{dx^2}-\alpha c_\ast^2(x)-d_g$
in the domain $\mathcal{D}(N)=\{v\in H^2(0,1)\colon v'(0)=0,\, v'(1)=0\}$. For abbreviation,   we write
$\mathcal{A}$ instead of $A+B(u^*)$, and consider the operator
\begin{equation*}
    \mathcal{A}=\begin{pmatrix}
      l(x) & m(x) & 0\\
      n(x) & -\varkappa & q(x)\\
      r(x) & \nu &  N
    \end{pmatrix}
\end{equation*}
with
\begin{gather*}
\varkappa=d_b+\nu,\quad n= 2\alpha c_\ast g_\ast,  \quad q= \alpha c_\ast^2,\\
    l= \frac {\theta b_\ast^2}{(b_\ast+ c_\ast)^2} - d_c,\quad m= \frac {\theta  c_\ast^2}{(b_\ast+ c_\ast)^2}, \quad
    r= \frac{k\beta c_\ast^{k-1}}{(1+c_\ast^k)^2} - 2\alpha c_\ast g_\ast, 
\end{gather*}
where $k\geq 2$. Of course, constant $\varkappa$ and functions $m$, $n$, $q$ are positive in $[0,1]$, since the stationary solution $(c_\ast, b_\ast, g_\ast)$ is positive.

The matrix $\cA$ is regarded as an unbounded operator in the space $X$ with the domain $\mathcal{D}(\cA)=C[0,1]\oplus C[0,1]\oplus\mathcal{D}(N)$. It can be written in the form $\cA=\mathcal{A}_0+\mathcal{A}_1$, where
\begin{equation*}
    \mathcal{A}_0=\begin{pmatrix}
      l  & m  & 0\\
      0  & -\varkappa & q\\
      0 & 0 &  N
    \end{pmatrix} \quad \text{and} \quad
    \mathcal{A}_1=\begin{pmatrix}
      0 &0 & 0\\
      n & 0 & 0 \\
      r  & \nu &  0
    \end{pmatrix}.
\end{equation*}
Notice that $\cA_1$ is a bounded operator in $X$.

\begin{theorem}\label{the_spectr_A0}
Let $h^*=\max\left\{-\varkappa, \lambda_1, \max_{x\in [0,1]}l(x)\right\}$, where $\lambda_1$ is the largest eigenvalue of $N$.
The spectrum of $\cA_0$ lies in the set $\{\lambda\in \mathbb{C}\colon \Rel \lambda \leq h^*\}$ and $h^*<0$.
\end{theorem}
\begin{proof}
To analyse the spectrum of $\cA_0$ we consider  
\begin{equation}\label{eqM}
\begin{cases}
(l-\lambda)\phi_1 +  m \phi_2= f_1,&\\
   -(\varkappa+\lambda) \phi_2 + q \phi_3= f_2, &\\
  N \phi_3 - \lambda \phi_3= f_3, &  
\end{cases}
\end{equation}
for each $f=(f_1, f_2, f_3)\in X= C[0,1]\times C[0,1]\times L^2(0,1)$.
Since the last equation in   \eqref{eqM}  is independent of $\phi_1$ and $\phi_2$,  for all $f_3\in L^2(0,1)$ and all $\lambda \notin \sigma(N)$ we have a unique solution $\phi_3 \in H^2(0,1)$.   Notice that due to the Sobolev embedding theorem $\phi_3 \in C[0,1]$.
 
If $\lambda\neq -\varkappa$,  $\lambda \notin \sigma(N)$, and  $l\notin R(l)$, then  we have a unique  solutions   of \eqref{eqM} in $C[0,1]\times C[0,1] \times H^2(0,1)$,  given by  
\begin{equation}\label{SolutionM}
\begin{aligned}
 \phi_1(x,\lambda)&=\frac {f_1(x)}{l(x)-\lambda}+\frac {m(x)(f_2(x)- q(x)\phi_3(x))}{(\varkappa+\lambda)(l(x)-\lambda)},\\
 \quad\phi_2(x,\lambda)&= -\frac {f_2(x)- q(x)\phi_3(x)}{\varkappa+\lambda}, \\
 \phi_3(x, \lambda) &= (N-\lambda)^{-1} f_3.
 \end{aligned}
\end{equation}
Here     $R(l)$ denotes the range of the continuous function  $l$  on $[0,1]$.  Thus  $\lambda$ is a point of  the resolvent set of $\cA_0$. 

If  either  $\lambda=-\varkappa$, $\lambda \in \sigma( N )$,  or  $\lambda \in R(l)$,  then the operator
$N$ does not have a solution for some $f_3\in L^2(0,1)$  or $\phi_1$ and $\phi_2$, defined in \eqref{SolutionM}, are not continuous for some $f_1, f_2 \in C[0,1]$. 

%system \eqref{eqM}  with  $\lambda = \lambda_0$ does not have a  solution in $C[0,1]\times C[0,1] \times H^2(0,1)$ for $f\in X$. 

Hence, the spectrum of $\cA_0$ coincides with the set
  $$
  \{-\varkappa\}\cup\{\lambda_j\}_{j=1}^\infty\cup  R(l),
  $$
where $\lambda_j$ are eigenvalues of the Sturm-Liouville operator
$N$ subject to the Neumann boundary conditions.

Since $\alpha c_\ast^2+d_g>0$ in $[0,1]$, all eigenvalues $\lambda_j$  of 
$N=\dfrac 1 \gamma \frac{d^2 }{dx^2}-\alpha c_\ast^2(x)-d_g$ subject to the Neumann boundary conditions are strictly negative, and hence
$\sigma(N)\subset \{\lambda\in \mathbb{C}\colon \Rel \lambda \leq \lambda_1\}$, where $\lambda_1<0$. 

Moreover,
the first equation in \eqref{eq_stead_state} yields
\begin{equation}\label{estimateL}
    l(x)=\frac {\theta b_\ast^2(x)}{(b_\ast(x)+ c_\ast(x))^2} - d_c< \frac {\theta b_\ast(x)}{b_\ast(x)+ c_\ast(x)} - d_c=-\frac{\mu}{c_*(x)}<0
\end{equation}
for all $x\in [0,1]$. Hence  $h^*$ is a negative number.
\end{proof}

Next, we will consider the operator $\cA$ as the perturbation of $\cA_0$  by the operator $\cA_1$.
Let $T$ and $S$ be operators with the same domain space $\cH$ such that $\mathcal{D}(T) \subset \mathcal{D}(S)$ and
\begin{equation}\label{Tbounded}
||Su|| \leq a ||u|| + b ||Tu||,   \qquad u\in \mathcal{D}(T) ,
\end{equation}
where $a$, $b$ are nonnegative constants. Then, we say that $S$ is relatively bounded with respect to $T$ or simply $T$-bounded.
Assume that $T$ is closed and  there exists a bounded operator $T^{-1}$,  and $S$ is $T$-bounded with constants $a$, $b$ satisfying the inequality
$$
  a||T^{-1}|| + b <1.
$$
Then, $T+S$ is a closed and bounded invertible operator \cite[\textit{Th.1.16}]{Kato}.

\begin{theorem}\label{the_spectr_A}
Under the assumptions of  Theorem~\ref{the_spectr_A0}
the spectrum of $\cA$ lies in the set $\{\lambda\in \mathbb{C}\colon \Rel \lambda \leq h\}$ with $h<0$ provided
\begin{equation}\label{estH2}
\begin{aligned}
  &  \max\left\{ \nu,  4\alpha \max\limits_{x\in [0,1]}  c_\ast(x) g_\ast(x),    \max\limits_{x\in [0,1]} \frac{k \beta 
c^{k-1}_\ast(x)}{(1+c_\ast^k(x))^2} \right\}<\left[ C_1 |h^\ast|^{-1/2} + C_2 |h^*|^{-1} \right]^{-1}, 
\end{aligned}
\end{equation}
where  $h^\ast$ as in Theorem~\ref{the_spectr_A0} and 
$$
\begin{aligned}
C_1 =&\frac{ \alpha} D \max\limits_{x\in [0,1]} (c_\ast(x))^2 + \frac{\alpha \theta}{D^2} \max\limits_{x\in [0,1]}\dfrac{  (c_\ast(x))^4 }{(b_\ast (x)+ c_\ast(x))^2}, \\
C_2 =& \max\left \{1+ \frac\theta D \max\limits_{x\in [0,1]} \dfrac{ ( c_\ast(x))^2 }{(b_\ast (x)+ c_\ast(x))^2},  \, 1+ C_1 \right \}, \\
D =& \min\left \{ \min\limits_{x\in [0,1]} \left|d_b + \nu  + \frac{\theta (b_\ast(x))^2}{(b_\ast (x)+ c_\ast(x))^2}- d_c\right|, \, \min\limits_{x\in [0,1]}\left |d_b + \nu - d_g - \alpha (c_\ast(x))^2\right|, \right. \\
&\left.  \qquad
\min\limits_{x\in [0,1]} \left|\frac{\theta (b_\ast(x))^2}{(b_\ast (x)+ c_\ast(x))^2} + d_g + \alpha (c_\ast(x))^2- d_c\right| \right\}.
\end{aligned}
$$
\end{theorem}
\begin{proof} Let $h\in(h^*,0)$. Choose  $\lambda_0\in \{\lambda\in \mathbb{C}\colon \Rel \lambda > h\}$.
Due to Theorem~\ref{the_spectr_A0} the operator $\cA_0-\lambda_0E$ is bounded invertible.
On the other hand,  for $\cA_1$ we have the following estimate % is bounded in $X^s$ with $s=1/2$, and
$$
\|\mathcal{A}_1 u\|_X \leq \|n \phi_1\|_{C[0,1]}+ \|r \phi_1\|_{L^2(0,1)} +
\|\nu \phi_2\|_{L^2(0,1)}  \leq a \|u\|_{X},
$$
where
$$
a=\max\left\{ \max_{x\in [0,1]}\big[|r(x)|+n(x)\big],\, \nu \right\}.
$$
Thus $\mathcal{A}_1$ is $(\cA_0-\lambda_0E)$-bounded with $b=0$ in the corresponding inequality \eqref{Tbounded}.
Consequently if $a\|(\cA_0-\lambda_0E)^{-1}\|_X<1$, then
the operator $\cA-\lambda_0E=\cA_0+\cA_1-\lambda_0E$ is bounded invertible. Using the definition of the operator $\cA_0$, estimates for $(N-\lambda_0)^{-1}$ and equations \eqref{SolutionM} we obtain 
\begin{equation}\label{estim_resolv}
\begin{aligned}
\|\phi_3\|_{L^2(0,1)} \leq& \;  \frac{\|f_3\|_{L^2(0,1)}}{\textrm{dist}(\lambda_0, \sigma(\mathcal{N}_0))},\\
 \|\nabla \phi_3\|_{L^2(0,1)} \leq & \; \frac{\|f_3\|_{L^2(0,1)}}{\textrm{dist}(\lambda_0, \sigma(\mathcal{N}_0))^{1/2}},\\
\|\phi_2\|_{C[0,1]}  = & \; \frac{\|f_2\|_{C[0,1]}}{\varkappa+\lambda_0} \\
&+  \max_{(0,1)} |q(x)| \frac{\|f_3\|_{L^2(0,1)}}{(\varkappa+\lambda_0) }
\left( \frac 1{\textrm{dist}(\lambda_0, \sigma(\mathcal{N}_0)) } +  \frac 1{\textrm{dist}(\lambda_0, \sigma(\mathcal{N}_0))^{1/2} } \right), \\
\| \phi_1\|_{C[0,1]} \leq & \;  \frac {\|f_1\|_{C[0,1]}}{\min\limits_{x\in[0,1]}|l(x)-\lambda_0|} +\frac {\max_{x\in[0,1]}|m(x)| \, \|f_2(x)\|_{C[0,1]}}{(\varkappa+\lambda_0) \min\limits_{x\in[0,1]}|l(x)-\lambda_0| } \\
 & + \frac {\max_{x\in [0,1]}|m(x) q(x)| }{\min\limits_{x\in[0,1]}|l(x)-\lambda_0| }\frac{\|f_3\|_{L^2(0,1)}}{(\varkappa+\lambda_0) }
\left[ \frac 1{\textrm{dist}(\lambda_0, \sigma(\mathcal{N}_0)) } +  \frac 1{\textrm{dist}(\lambda_0, \sigma(\mathcal{N}_0))^{\frac 12} } \right].  \end{aligned}
\end{equation}
Hence we obtain that 
\begin{equation*}
    \|(\cA_0-\lambda_0E)^{-1}\|_X \leq 
  \frac {\tilde C_1}{\textrm{dist}\,(\lambda_0, \sigma(\mathcal{A}_0))^{1/2}}+ \frac {\tilde C_2}{\textrm{dist}\,(\lambda_0, \sigma(\mathcal{A}_0))},  
\end{equation*}
where $$
\begin{aligned}
&\tilde C_1 =  \max\limits_{x\in [0,1]}|q(x)|\tilde D^{-1}+  \max\limits_{x\in [0,1]}|m(x)\,q(x)| \tilde D^{-2}, \\
&\tilde C_2 = \max\left \{1+ \tilde D^{-1}\max\limits_{x\in [0,1]}|m(x)|,\,  1+ \tilde C_1 \right \}, \\
&\text{ with } \tilde D = \min\{ \min\limits_{x\in [0,1]} |\varkappa + l(x)|, \, \min\limits_{x\in [0,1]} |\varkappa - d_d - q(x)|, \; 
\min\limits_{x\in [0,1]} |l(x) + d_d + q(x)| \}.
\end{aligned}
$$

If \eqref{estH2} holds, then $a< [\tilde C_1/|h^\ast|^{1/2}+\tilde C_2/|h^*| ]^{-1}$ since $|r(x)|+n(x)=r+n=\dfrac{k\beta c_\ast^{k-1}}{(1+c_\ast^k)^2}$ for $r(x)\geq 0$
 and $|r(x)|+n(x)\leq 2 n(x)$ for $r(x)<0$, $x\in[0,1]$,
where $n=2\alpha c_*g_*$. Therefore $a<\big [\tilde C_1/|h^\ast-h|^{1/2}+\tilde C_2/|h^*-h| \big]^{-1}$  for some $h\in(h^*,0)$ and $a\|(\cA_0-\lambda_0E)^{-1}\|<1$. We showed that the set
$\{\lambda\in \mathbb{C}\colon \Rel \lambda > h\}$ belongs to the resolvent set of $\cA$, which completes the proof.
\end{proof}

%%%%%%%%%%%%%%%%%%%%%%%%%%%%%%%%%%%%%%%%%%%%%%%%%%%%%%%%%%%%%%%%5

\subsection{Linearized stability conditions}\label{Numeric}

Theorem \ref{the_spectr_A} implies that under the assumptions \eqref{estH2}   the spectrum of linearized operator
$\mathcal A$ lies in $\{\lambda \in \mathbb C: {\Rel } \, \lambda < h\}$ for some  $h<0$.
Then, from Proposition \ref{Stabil2} follows the linearized stability of  a stationary  solution $(c_*,b_*,g_*)$  in
$X^s\subseteq C[0,1]\oplus C[0,1]\oplus H^{2s}(0,1)$, $s\in [1/2,1)$.
\begin{corollary}
A spatially nonconstant stationary solution $(c_*,b_*,g_*)$ of system (\ref{main}) is li\-nearly stable in $X^s$  if the following conditions are fulfilled
\begin{eqnarray} \label{stab_ass}
\begin{aligned}
 k\beta  \frac{c_\ast^{k-1}(x)}{(1+c_\ast^k(x))^2} &< [C_1|h^\ast|^{-1/2} + C_2|h^\ast|^{-1}]^{-1}, \quad \\
 \nu &< [C_1|h^\ast|^{-1/2} +C_2|h^\ast|^{-1}]^{-1}, \\
 4\alpha\,  c_\ast(x) \, g_\ast(x) & < [C_1|h^\ast|^{-1/2} + C_2|h^\ast|^{-1}]^{-1}, 
\end{aligned}
\end{eqnarray}
for all $x\in[0,1]$, where $h^\ast$ is defined  in Theorem \ref{the_spectr_A0} and the constants $C_1$ and $C_2$ are defined in Theorem~\ref{the_spectr_A}.
\end{corollary}

\begin{remark} It is enough to verify  \eqref{stab_ass} in a  local minimum  $g_0$ of the potential $U$ from Theorem \ref{exist_stat} and for the  sets of   parameters such that the local minimum exists.
Then,  Corollary \ref{small_amp} and the strict inequality in \eqref{stab_ass} provide stability of
  a stationary solution with the small amplitude.
\end{remark}

\section{Conclusions}

In this paper we performed linearized stability analysis of the nonhomogeneous stationary solutions of  a single reaction-diffusion equation coupled to ordinary differential equations. We focused on a specific model of pattern formation in a system of cells with the proliferation regulated by a diffusive growth factor. 
We extended results on the stationary problem obtained in \cite{Marciniak} to the larger space of parameters and showed existence of inifinitely many stationary solutions. Numerical calculations indicate that
the relationship between $d_c$ and $\theta$ is signi\-fi\-cant  for the existence of nonconstant stationary solutions. 
It is observed that for $d_c \neq \theta$ the system~\eqref{main} admits
the nonconstant stationary solutions when the diffe\-rence $|d_c-\theta|$ is small enough.
For $d_c=\theta$ the value of the parameter $\beta$ plays an important role in   existence of the stationary solutions.

Performing spectral analysis using sectorial operators and the perturbation theo\-ry we obtained conditions for the linearized stability of the
 spatial patterns.  Formulation of the conditions in terms of the model parameters independent of the values of 
stationary solutions 
is difficult to obtain due to the complex structure  of the stationary problem. For some parameter values, due
 to existence of multiple solutions of the ODEs subsystem, the model exhibits also discontinuous 
stationary solutions and it is difficult to decide which patterns we see in numerical simulations.  
The mechanism of pattern selection is a topic of further analytical and numerical studies.

%%%%%%%%%%%%%%%%%%%%%%%%%%%%%%%%%%%%%%%%%%%%%%
%%%%%%%%%%%%%%%%%%%%%%%%%%%%%%%%%%%%%%%%%%%%%%%%%%%%%

% \medskip
% Received March   2010; revised January  2011.
 
% \medskip

\end{document}